\documentclass{article}
\usepackage{geometry}                
\usepackage{graphicx}
\usepackage{amssymb}
\usepackage{epstopdf}
\usepackage{multirow}
\usepackage{color}
\usepackage[lofdepth,lotdepth]{subfig}
\newtheorem{theorem}{Theorem}

\newtheorem{lemma}[theorem]{Lemma}

\def\QED{\ensuremath{{\square}}}
\def\markatright#1{\leavevmode\unskip\nobreak\quad\hspace*{\fill}{#1}}
\newenvironment{proof}
  {\begin{trivlist}\item[\hskip\labelsep{\bf Proof.}]}
  {\markatright{\QED}\end{trivlist}}

\begin{document}

\title{Robust Non-Parametric Data Approximation of Pointsets via Data Reduction\thanks{Work was supported in part by the Natural Science and Engineering Research Council of Canada (NSERC)}}

\author{Stephane Durocher\thanks{University of Manitoba, Winnipeg, Canada, {\tt durocher@cs.umanitoba.ca}, {\tt alex\_{}leblanc@umanitoba.ca}, {\tt jason\_{}morrison@umanitoba.ca} and {\tt mskala@cs.umanitoba.ca}} \and Alexandre Leblanc\footnotemark[2] \and Jason Morrison\footnotemark[2] \and Matthew Skala\footnotemark[2]}

\maketitle

\begin{abstract}
In this paper we present a novel non-parametric method of simplifying
piecewise linear curves and we apply this method 
as a statistical approximation of structure within sequential data in the plane.
We consider the problem of minimizing the average length of sequences of
consecutive input points that lie on any one side of the simplified curve.
Specifically, given a sequence $P$ of $n$ points in the plane that determine a 
simple polygonal chain consisting of $n-1$ segments, 
we describe algorithms for selecting an ordered subset $Q \subset P$
(including the first and last points of $P$) that determines a second
polygonal chain to approximate $P$,
such that the number of crossings between the two polygonal chains is maximized,
and the cardinality of $Q$ is minimized among 
all such maximizing subsets of $P$.
Our algorithms have respective running times 
$O(n^2\log n)$ when $P$ is monotonic and
$O(n^2\log^2 n)$ when $P$ is an arbitrary simple polyline.
Finally, we examine the application of our algorithms iteratively in a 
bootstrapping technique to define a smooth robust non-parametric approximation of the original sequence.
\end{abstract}

\section{Introduction}

Given a simple
polygonal chain $P$ (a \emph{polyline}) defined by a sequence of points
$(p_1,p_2,\ldots,p_n)$ in the plane, the \emph{polyline simplification problem} is
to produce a simplified polyline $Q=(q_1,q_2,\dots,q_k)$, where $k < n$.
The polyline $Q$ represents an approximation of $P$
that optimizes one or more objectives evaluated as functions of $P$ and $Q$.
For $P$ to be \emph{simple}, the points $p_1,\ldots,p_n$ must be
distinct and $P$ cannot intersect itself.

Motivation for studying polyline simplification comes from the
fields of computer graphics and cartography, 
where simplification is used to render
vector-based features such as
streets, rivers, or coastlines onto a screen or a map at appropriate
resolution with acceptable error \cite{agarwal:2002},
as well in problems involving computer animation, pattern matching and  
geometric hasing (see survey by Alt and Guibas for details \cite{alt:1999}).
Our present work removes the arbitrary parameter
previously required to describe acceptable error between $P$ and $Q$, and
provides a simplification method that is robust to some forms of noise.

Typical polyline simplication algorithms require that ``distance'' between
two polylines be measured using a function denoted here by $\zeta(P,Q)$. 
The specific measure of interest differs depending on the focus of the
particular problem or article; however, three measures are popular:
Chebyshev error $\zeta_C$, Hausdorff distance $\zeta_H$, and Fr\'{e}chet
distance $\zeta_F$.  In informal terms, the Chebyshev error is the maximum
absolute difference between $y$-coordinates of $P$ and $Q$ (maximum
residual); the symmetric Hausdorff distance is the distance between the most
isolated point of $P$ or $Q$ with respect to the other polyline; and the
Fr\'{e}chet distance is more complicated, being the shortest possible
maximum distance between two particles each moving forward along $P$ and
$Q$.  Alt and Guibas give more formal definitions \cite{alt:1999}.  We define
a new measure of quality or similarity, to be maximized, rather than using
an ``error'' to be minimized.  Our crossing measure is a combinatorial
description of how well $Q$ approximates $P$.  It is invariant under a
variety of geometric transformations of the polylines, and is often robust to
uncertainty in the locations of individual points.

Previous work on polyline simplification is generally divided into
four categories depending on what property
is being optimized and what restrictions
are placed on $Q$~\cite{alt:1999}. Problems can be classified as those
that either
require an approximating polyline $Q$ having the minimum number of segments 
(minimizing $|Q|$) for a
given acceptable error $\zeta(P,Q) \leq\epsilon$, or a $Q$ with minimum
error $\zeta(P,Q)$ for a given value of $|Q|$.  These are called min-$\#$
problems and min-$\epsilon$ problems respectively.  These two types of
problems are each further divided into ``restricted'' problems where the
points of $Q$ are required to be a subset of those in $P$ and to include the
first and last points of $P$ ($q_1 = p_1$ and $q_k = p_n$), 
and ``unrestricted'' problems, where the
points of $Q$ may be arbitrary points on the plane.  Under this classification,
the polyline simplification $Q$ we examine is a restricted min-$\#$ problem
for which a subset of points of $P$ is selected (including $p_1$ and $p_n$)
where the objective measure
$\zeta(P,Q)$ is the number of crossings between $P$ and $Q$ 
and an optimal simplification first
maximizes (rather than minimizing) the crossing number and then has a
minimum $|Q|$ given the maximum crossing number.  

While the restricted min-$\#$ problems find the smallest sized approximation within a 
given error $\epsilon$, an earlier approach was to find any approximation within the 
given error.  The cartographers Douglas and Peucker \cite{douglas:1973fk}
 developed a heuristic algorithm
where an initial (single segment) approximation was evaluated and the furthest point was then added to the simplification.  This technique remained inefficient until series of papers by Hershberger and Snoeyink  concluded that the problem could be solved in $O(n\log^* n)$ time and linear space \cite{hershberger:1998uq}.

The most relevant previous
literature is on restricted min-$\#$ problems. Imai and Iri~\cite{imai:1988fk} presented 
an early solution to the
restricted polyline simplification problem using $O(n^3)$ time and
$O(n)$ space.  The version they study optimizes $k=|Q|$ while maintaining
that the Hausdorff metric between $Q$ and $P$ is less than the parameter
$\epsilon$.  Their algorithm was subsequently improved by Melkman and
O'Rourke~\cite{melkman:1988uq} to $O(n^2\log n)$ time and then by Chan and
Chin~\cite{chan:1992fk} to $O(n^2)$ time.  Subsequently, Agarwal and
Varadarajan~\cite{agarwal:2000kx} changed the approach from finding a
shortest path in an explicitly constructed graph to an implicit method that
runs in $O(f(\delta)n^{\frac{4}{3}+\delta})$ time.  Agarwal and Varadarajan
used the $L_1$ Manhattan and $L_\infty$ Chebyshev metrics instead of the
previous works' Hausdorff metric.  Finally, Agarwal \emph{et al.}~study a
variety of metrics and give approximations of the min-$\#$ problem in $O(n)$
or $O(n\log n)$ time.

Our algorithm for minimizing $|Q|$ while optimizing our non-parametric
quality measure requires $O(n^2\log n)$ time when
$P$ is monotonic in $x$, or $O(n^2\log^2 n)$ time when $P$ is a
non-monotonic simple polyline on the plane, both in $O(n)$ space.
The near-quadratic times are
remarkably similar to the optimal times achieved in the parametric
version of the problem using Hausdorff distance
\cite{agarwal:2002,chan:1992fk},
suggesting the possibility that the problems may have similar complexities.

In the next section, we define the crossing measure $\chi(Q,P)$ and relate
the concepts and properties of $\chi(Q,P)$ to previous work in both
polygonal curve simplification and robust approximation.  
In Section~\ref{sect:method}, we describe our algorithms to compute
simplifications of monotonic and non-monotonic simple polylines that
maximize $\chi(Q,P)$.  Section~\ref{sect:monotonicresults} presents our
results in applying the method to $x$-monotonic polylines that model 2-D
functional (e.g., measured) data and describes the use of this
simplification method to approximate ``shape'' and ``noise'' without assuming
a parametric model for either.

\section{Crossing Measure}
\label{sect:measure}

The crossing measure $\chi(Q,P)$ is defined for a sequence of $n$ distinct
points $P=(p_1,p_2,\ldots,p_n)$ and a subsequence of $k$ distinct points $Q
\subset P, Q=(q_1,q_2,\ldots,q_k)$ with the same first and last values:
$q_1= p_1$ and $q_k = p_n$.  For each $p_i$ let $(x_i,y_i)=p_i \in
\mathbb{R}^2$.  To understand the crossing measure it is necessary to
introduce the idea of left and right sidedness of a point relative to a
directed line segment.  A point $p_j$ is on the left side of a segment
$S_{i,i+1}=[p_i,p_{i+1}]$ if the signed area of the triangle formed by the
points $p_i,p_{i+1},p_j$ is positive.  Correspondingly, $p_j$
is on the right side of the segment if the signed area is negative.  The
three points are collinear if the area is zero.

For any endpoint $q_i$ of a segment in $Q$ it is possible to determine the side
of $P$ on which $q_i$ lies. Since $Q$ is a polyline using a
subset of the points defining $P$, for every segment $S_{i,i+1}$ there
exists a corresponding segment of $S_ {\pi(j),\pi(j+1)}$ such that $\pi(j)
i<i+1\leq\pi(j+1)$.  The endpoints of $S_{\pi(j),\pi(j+1)}$ are given a side
based on $S_{i,i+1}$ and vice versa.  Two segments intersect if they share a
point.  Such a point is interior to both segments if only if both segments
change sides with respect to each other or the intersection is at an
endpoint of at least one endpoint is collinear to the other segment \cite[p. 
566]{skiena:2008uq}.  The \emph{crossing measure} $\chi(Q,P)$ is the number
of times that $Q$ changes sides from properly left
to properly right of $P$ due to an intersection between the polylines.  A
single crossing can be generated by any of five cases listed below (see Figure~\ref{fig:cases}):

\begin{enumerate}
  \item{A segment of $Q$ intersects $P$ at a point distinct from any
    endpoints;}
  \item{two consecutive segments of $P$ meet and cross at a point interior
    to a segment of $Q$;}
  \item{one or more consecutive segments of $P$ are collinear to the interior of a segment of $Q$ with the previous and following segments of $P$ on opposite sides of that segment of $Q$;}
\item{two consecutive segments of $P$ share their common point with two consecutive segments of $Q$ and form a crossing; or }
  \item{in a generalization of the previous case, instead of being a
    single point the intersection comprises one or more sequential segments
    of $P$ and possibly $Q$ that are collinear or identical.}
\end{enumerate}

\begin{figure}
\subfloat{
\includegraphics[scale=0.6]{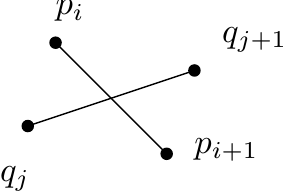}
}
\subfloat{
\includegraphics[scale=0.6]{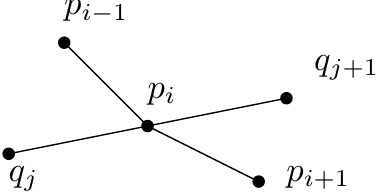}
}
\subfloat{
\includegraphics[scale=0.6]{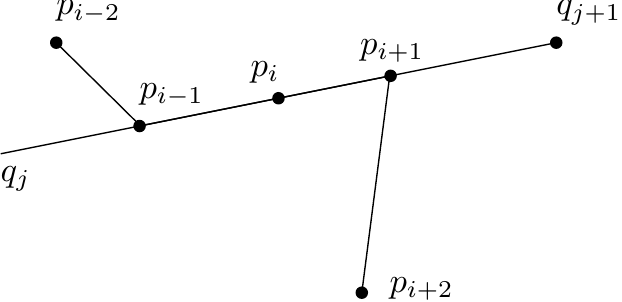}
}
\subfloat{
\includegraphics[scale=0.6]{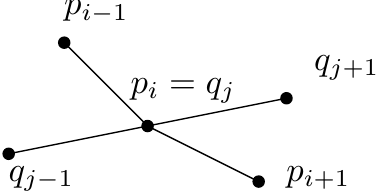}
}
\subfloat{
\includegraphics[scale=0.6]{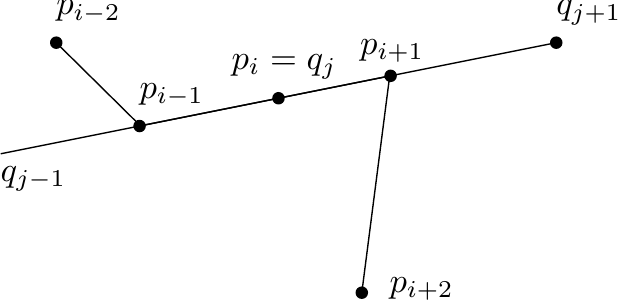}
}
\caption{Examples of the five cases generating a single crossing.}
\label{fig:cases}
\end{figure}

In Section~\ref{sect:xsingleseg} we discuss how to compute the crossings for the first three
cases, which are all cases where
crossings involve only one segment of $Q$.  The remaining cases involve more
than one segment of $Q$, because an endpoint of one segment of $Q$ or even
some entire segments of $Q$ are coincident with one or more segments of $P$;
those cases are discussed in Section~\ref{sect:xmultiplesegment}.
 	
\begin{figure}
\label{fig}
\subfloat{
\includegraphics[scale=0.6]{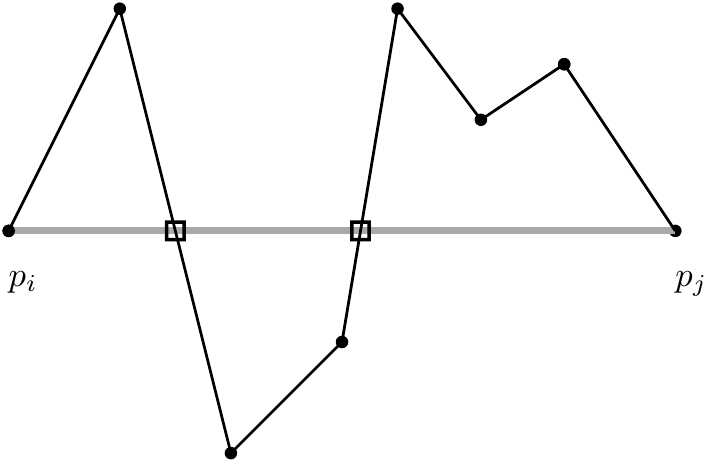}
}
\qquad
\subfloat{
\includegraphics[scale=0.6]{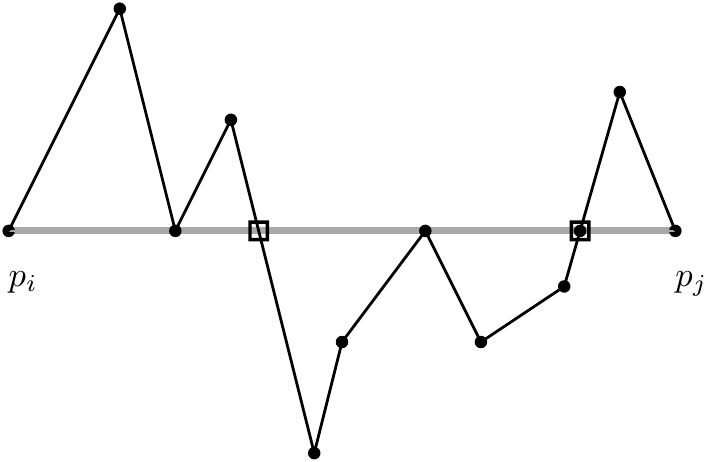}
}
\qquad
\subfloat{
\includegraphics[scale=0.6]{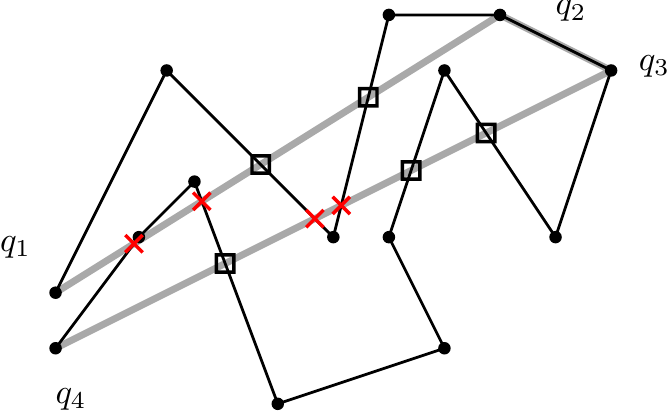}
}
\caption{Crossings are indicated with a square and false crossings are
  marked with a red x.  Crossing are only counted when a segment intersects
  the portion of the chain between its own endpoints.}
\end{figure}
	

In the case where the $x$-coordinates of $P$ are monotonic, $P$ describes a
function $Y$ of $x$ and $Q$ is an approximation $\hat{Y}$ of that function.  
The signs of the residuals $r = (r_1,r_2,\ldots,r_n)=Y_P - \hat{Y}$ 
are computed at the $x$-coordinates of $P$ and are equivalent to
the sidedness described above.  The crossing number is the
number of proper sign changes in the sequence of residuals.  The resulting
simplification maximizes the likelihood that adjacent residuals would have
different signs, while minimizing the number of original data points
retained conditional on that number of sign changes.  Note that if $r$ was
independently and identically selected at random from a distribution with
median zero, then any adjacent residuals in the sequence
$(r_1,r_2,\ldots,r_n)$ would have different signs with probability $1/2$.

\subsection{Counting Crossings With a Segment}
\label{sect:xsingleseg}

To compute a simplification $Q$ with optimal crossing number for a given $P$,
we consider the optimal numbers of crossings for segments of $P$ and combine
them in a dynamic programming algorithm.  Starting from a point $p_i$ we
compute optimal crossing numbers for each of the $n-i$ segments that start
at point $p_i$ and end at some $p_j$ with $i<j\le n$.  Computing all $n-i$
optimal crossing numbers for a given $p_i$ simultaneously in a single pass
is more efficient than computing them for each $(p_i,p_j)$ pair separately. 
These batched computations are performed for each $p_i$ and the results used
to find $Q$.

To compute a single batch we will consider the angular order of points in
$P_{i+1,n}=\{p_{i+1},\ldots,p_n\}$ with respect to $p_i$.  Let
$\rho_i(j)$ be a function on the indices representing the clockwise angular
order of points $p_j$ within this set, such that $\rho_i(j)=1$ for 
all $p_j$ having the smallest clockwise angle measured from the vertical
line passing through $p_i$, and $\rho_i(j)\le\rho_i(k)$ if and only if this
angle for $p_j$ is less than or equal to the corresponding angle for $p_k$.
See Figure~\ref{fig:angularchains}.  Using this
angular ordering we partition $P_{i+1,n}$ into chains and process the
batch of crossing number problems as discussed below.

We define a \emph{chain} with respect to $p_i$ to be a consecutive sequence
$P_{\ell,\ell '}\subset P_{i+1,n}$ with non-decreasing angular order.  That
is, either $\rho_i(l')\ge \rho_i(j+1) \ge \cdots \ge \rho_i(l)$ or
$\rho_i(l)\le \rho_i(j+1) \le \cdots \le \rho_i(l')$, with the added
constraint that chains cannot cross the vertical ray above $p_i$.  Each
segment that does cross is split into two pieces using two ``artificial''
points on the ray per crossing segment.  The points on the ``low'' segment
portions have rank $\rho_i=1$ and the identically placed other points have
rank $\rho_i=n+1$.  These points do not increase the complexities by more
than a constant factor and are not mentioned again unless specifically
required.  Processing $P_{i+1,n}$ into its chains is done by first computing
the angle from vertical for each point and storing that information with the
points.  Then the points are sorted by angular order around $p_i$ and
$\rho_i(j)$ is computed as the rank of $p(j)$ in the sorted list.  Since
this algorithm works in the real RAM model, this step can be done in
$O(n\log n)$ time with linear space to store the angles and ranks.  Creating
a list of chains is then computable in $O(n)$ time and space by storing the indices of the
beginning and end of each chain encountered while checking points $p_j$ in
increasing order from $j=i+1$ to $j=n$.  The process to identify all chains involves two steps.
First all
segments are checked to determine if they intersect with the vertical ray,
each in $O(1)$ time.  Such an intersection implies that the previous chain
should end and the segment that crosses the ray should be a new chain (note an artificial index of $i+\frac{1}{2}$ can denote the point that crosses the vertical).  The
second check is to determine if the most recent segment has a different
angular direction from the previous segment.  If so, the previous chain ``ended''
with the previous point and the new chain ``begins'' with the current segment. 
Each chain is oriented from lowest angular order to highest angular
order.

\begin{figure}
\includegraphics[scale=0.5]{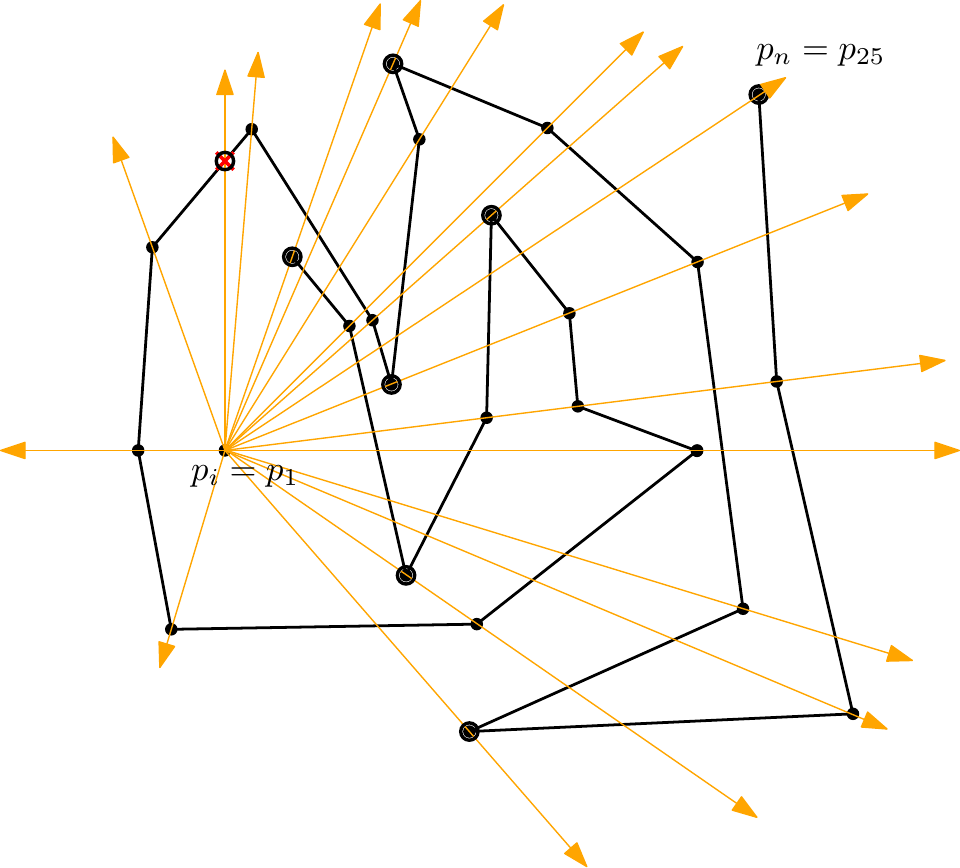}
\caption{An example of the angular order of vertices in $P_{i+1,n}$ and the
  resulting chains}
\label{fig:angularchains}
\end{figure}
	
\begin{lemma}\label{lem:segmentcross}
	
Consider any chain $P_{\ell,\ell'}$ (wlog\ assume $\ell<\ell'$). With respect
to $p_i$ the segment $S_{i,j}: (i<j\leq n)$ can have at most one crossing
strictly interior to $P_{\ell,\ell'}$.
	
\end{lemma}

\begin{proof} Three cases need to be considered.
	
{\bf Case 1:} $\rho_i(\ell)=\rho_i(j)$ or $\rho_i(j)=\rho_i(\ell')$. Note
that if $\ell=n$ then no crossing can exist because at least one end (or all)
of $P_{k,\ell}$ is collinear with $S_{i,j}$ and no proper change in sidedness
can occur in this chain to generate a crossing.
		
{\bf Case 2:} $\rho_i(j) \notin [\rho_i(\ell)$, $\rho_i(\ell')]$.  These
cases have no crossings with the chain because $P_{k,\ell}$ is entirely on
one side of $S_{i,j}$.  A ray exists
between either $\rho_i(j)<\rho_i(\ell)$ or $\rho_i(m)<\rho_i(j)$ that separates
$P_{k,l}$ from $S_ {P,i,j}$ and thus no crossings can occur between the
segment and the chain.
		
{\bf Case 3:} $\rho_i(j) \in (\rho_i(\ell),\rho_i(\ell'))$.
Assume that the chain causes at least two crossings.  Pick the lowest index segment for each of the two crossings that are the fewest segments away from $p_i$.
By definition there are no crossings of segments between these two segments.
Label the point with lowest index of these two segments $p_{\lambda}$ and
the point with greatest index $p_{\lambda'}$.  Define a possibly degenerate cone $Phi$ with a base $p_i$ and rays through 
$p_{\lambda}$ and $p_{\lambda'}$.  This cone, by definition, separates the segments from $p_{\lambda + 1}$ to $p_{\lambda'-1}$ from the remainder of the chain.  Since this sub-chain cannot circle $p_i$ entirely there must exist one or more points that have a maximum (or minimum) angular index, which is a contradiction to the definition of the chain.  Hence there must be zero or one crossings only.\end{proof}
	
The algorithm for computing the crossing measure on a batch of segments
depends on the nature of $P$.  If $P$ is
$x$-monotone, then the chains can be ordered by increasing $x$-coordinates or
equivalently by the greatest index amoung the points that define them.  Then
a segment $S_{i,j}$ intersects any chain $P_{\ell,\ell'}$ exactly once if
its $x$-coordinates are less than $p_j$ and $\rho_i(j) \in
\left(\rho_i(\ell),\rho_i(\ell')\right)$ (i.e., Case 3 of
Lemma~\ref{lem:segmentcross}).  The algorithm maintains a modified segment
tree with one angular order interval per chain previously included, using
the modified segment tree described by van~Kreveld~\emph{et~al.}
\cite[p.~237]{berg:2008fk}.  This data structure requires $O(\log n)$ time
per insertion and $O(n)$ space.  Each point's crossing number is queried in
$O(\log n)$ time, with points examined in order of increasing indices. 
Once each chain's points have all been queried the chain's interval is
added.  Correctness follows from the fact that no segment considered
can have a crossing within any chain it ends, and chains that span a
point's angular order intersect once if the point is sufficiently distant
from $p_i$ relative to the chain.  These facts are guaranteed by
$x$-monotonicity and the proof of Lemma~\ref{lem:segmentcross}.
	
The problem becomes more difficult
if we assume that $P$ is simple but not necessarily monotonic in $x$.
While chains describe angular order quite
nicely, the non-monotone nature of $P$ does not allow a consistent
implicit ordering of chain boundaries.  Thus queries will be of a specific
nature: for a given point $p_j$, we must determine how many chains are
closer to $i$ and have a lower maximum index than $j$.  Note that chains do
not cross and can only intersect at their endpoints due to the
non-overlapping definition of chains and the simplicity of $P$.  Therefore,
sweeping a ray from $p_i$, initially vertical, in increasing $\rho_i$ order
defines a partial order on chains with respect to their distance from $p_i$. 
Using a topological sweep~\cite[p.~481]{skiena:2008uq} it is possible to
determine a unique order that preserves this partial ordering of chains. 
Since there are $O(n)$ chains and changes in ``neighbours'' defining the
partial order occur at chain endpoints, there are $O(n)$ edges in the partial
order and this operation requires $O(n\log n)$ time to determine the events
in a sweep and $O(n)$ time to compute the topological ordering.  Without
loss of generality assume that the chains closest to $p_i$ have a lower
topological order.

In our algorithm each chain will be labelled with two labels: the maximum
index of its defining points and the topological order.  Furthermore, each
point $p_j$ will be labelled with the topological order of the chain to which 
it belongs (or the minimum of the two if it is in two chains).  A sweep in
increasing $\rho_i$ order maintains the set of chains whose range of angular
orders properly includes the current $\rho_i(j)$.  Thus to query the number
of crossings of $S_{i,j}$ we need to determine from the current set of
chains the number of them whose topological order is strictly less than the
chain or chains containing $p_j$ and whose maximum index is less than $j$. 
Querying the set in this way is an orthogonal range counting query in
$\mathbb{R}^2$, and such queries can be performed in $O(\log^2n)$ time and
$O(n)$ space with insertion and deletion of chains in $O(\log^2n)$ time per
event on an elementary pointer machine~\cite{chazelle:1988vn}.  The order of
operations is as follows: first, build the range counting structure by
inserting all chains that begin at the vertical ray from $p_i$; next, for
each unique angular order delete all segments whose maximum $\rho_i$ order
is achieved (this maintains the proper intersection of $\rho_i$ previously
mentioned); compute the crossing number of all points with this angular
order by querying the data structure; finally, add any new chains starting
at this $\rho_i$.  The artificial points on the vertical are not queried. 
Correctness follows from Lemma~\ref{lem:segmentcross} and the
previous discussion.

\subsection{Crossings Due to Neighbouring Simplification Segments}
\label{sect:xmultiplesegment}

There are two cases of a crossing being generated that involve more than one segment of $Q$ and it
is these cases we address now.  Suppose that $p_i=q_j$.  Then there is
an intersection between $P$ and $Q$ at this point, and we must detect if a
change in sidedness accompanies this intersection.  Assume initially
that $P$ does not contain any consecutively collinear segments; we will
consider the other case later.  

We begin with the non-degenerate case where $(p_{i-1},p_{i+1},q_{i-1},q_{i+1})$ 
are all distinct points. Each of the
points $q_{j-1}$ and $q_{j+1}$ can be in one of four locations: in the cone
left of $(p_{i-1},p_i,p_{i+1})$; in the cone right of
$(p_{i-1},p_i,p_{i+1})$; on the ray defined by $S_{i,i-1}$; or on the ray
defined by $S_{i,i+1}$.  These are labelled in Figure~\ref{fig:locations} as 
regions $I$ through $IV$ respectively.  In Cases $III$ and $IV$ it
may also be necessary to consider the location of $q_{i-1}$ or
$q_{j+1}$ with respect to $S_{i-2,i-1}$ or $S_{i+1,i+2}$.

Within the degenerate ``case'' where the points may not be unique: if  $p_i = q_j$ and $p_{i+1} \neq q_{j+1}$, then any change in sidedness is handled at $p_i$ and can be detected by verifying the
previous side from the polyline.  If, however, $p_{i+1} = q_{j+1}$, then any change in
sidedness will be handled further along in the simplification.
	
\begin{figure}
\includegraphics{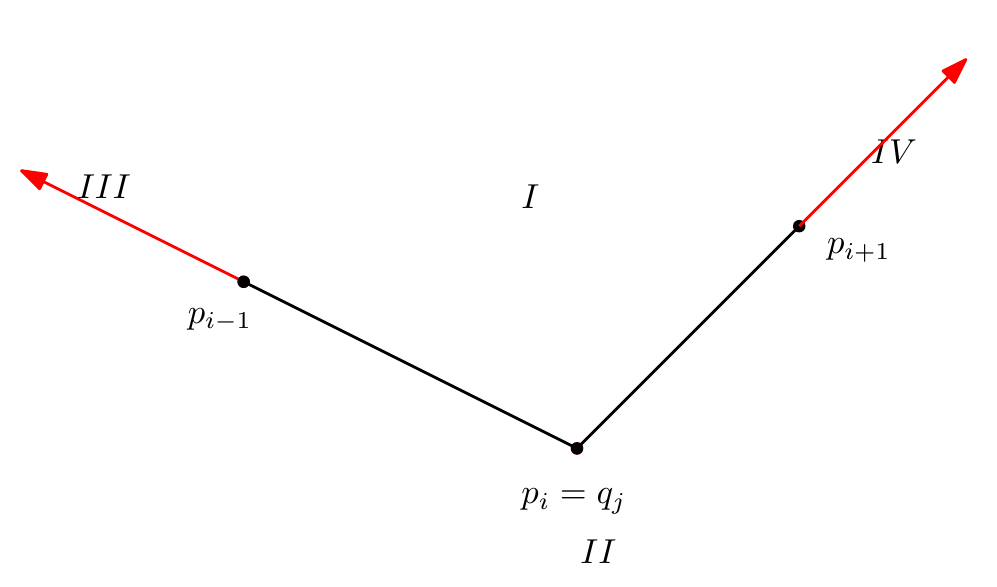}
\caption{Regions around $p_j$ that determine a crossing at $p_j$}
\label{fig:locations}
\end{figure}

By examining these points it is possible to assign a sidedness to
the end of $S_{\pi(j-1),\pi(j)}$ and the beginning of
$S_{\pi(j),\pi(j+1)}$.  Note that the sidedness of a point $q_{j-1}$ with
respect to $S_{i-2,i-1}$ can be inferred from the sidedness of $p_{i-2}$ with
respect to $S_{\pi(j-1),i-1}$, and that property is used in the case of
regions $III$ and $IV$. 
The assumed lack of consecutive collinear segments requires that
$\{p_{i-2},p_{i+2}\}\in I \cup II$ and thus Table~\ref{tab:leftright} is a
complete list of the possible cases when $|P|\geq 5$.
For cases involving $III$ or $IV$ where $i\notin[3,n-2]$ then the case
is labelled collinear (we discuss the consequences of this choice later).

\begin{table}
\begin{center}
\begin{tabular}{|c|c|c|}
\hline
Entity & Categorization & Conditions \\
\hline
\multirow{8}{*}{End of $S_{\pi(j-1),i}$} & \multirow{2}{*}{collinear (1)} & $q_{j-1} = p_{i-1}$ \\
&& $q_{j+1} = p_{i+1}$ \\ 
\cline{2-3}
& \multirow{3}{*}{left (2)} & $q_{j-1}\in I$ \\
& &  $(q_{j-1}\in III) \wedge (p_{i-2}\in II)$ \\
& &  $(q_{j-1}\in IV) \wedge (p_{i+2}\in II)$ \\
\cline{2-3}
& \multirow{3}{*}{right (3)} & $q_{j-1}\in II$ \\
& &  $(q_{j-1}\in III) \wedge (p_{i-2}\in I)$ \\
& &  $(q_{j-1}\in IV) \wedge (p_{i+2}\in I)$\\
\hline
\multirow{8}{*}{Beginning of $S_{i,\pi(j+1)}$} & \multirow{2}{*}{collinear (1)} & $q_{j-1} = p_{i-1}$ \\
&& $q_{j+1} = p_{i+1}$ \\ 
\cline{2-3}
& \multirow{3}{*}{left (2)} & $q_{j+1}\in I$ \\
& &  $(q_{j+1}\in III) \wedge (p_{i-2}\in II)$ \\
& &  $(q_{j+1}\in IV) \wedge (p_{i+2}\in II)$ \\
\cline{2-3}
& \multirow{3}{*}{right (3)} & $q_{j+1}\in II$ \\
& &  $(q_{j+1}\in III) \wedge (p_{i-2}\in I)$ \\
& &  $(q_{j+1}\in IV) \wedge (p_{i+2}\in I)$\\
\hline
\end{tabular}
\caption{Left, right, and collinear labels applied to beginning or end of a segment at $p_j$}
\label{tab:leftright}
\end{center}
\end{table}

A single crossing occurs if and only if the end of $S_{\pi(j-1),i}$ is on
the left or right when the beginning of $S_{i,\pi(j+1)}$ is the opposite. 
Furthermore, the end of any simplification $Q_{1,j}$ of $P_{1,i}$ that ends
in $S_{\pi(j-1),i}$ is labelled left or right in the same way that the end
of $S_{\pi(j-1),i}$ is labelled.  This labelling is consistent with the
statement that the simplification last approached the polyline $P$ from the
side indicated by the labelling.  To maintain this invariant in the
labelling of the end of polylines, if $S_{\pi(j-1),i}$ is labelled as
collinear then the simplification $Q_{1,j}$ needs to have the same labelling
as $Q_{1,j-1}$.  As a basis case, the simplifications of $P_{1,2}$ and
$P_{1,1}$ are the result of the identity operation so they must be
collinear.  Note that a simplification labelled collinear
has no crossings.

The constant number of cases in Table~\ref{tab:leftright}, and the
constant complexity of the sidedness test, imply that we can compute the
number of crossings between a segment and a chain, and therefore the
labelling for the segment, in constant time.
Let $\eta(Q_{1,j},S_{\pi(j-1),i})$ represent the number of extra crossings
(necessarily $0$ or $1$) introduced at $p_j$ by joining $Q_{1,j}$ and
$S_{\pi(j-1),i}$.  We have $\chi(Q_{1,j},P_{1,i}) =
\chi(Q_{1,j-1},P_{1,\pi(j-1)}) + \chi(S_{\pi(j-1),i},P_{\pi(j-1),i}) +
\eta(Q_{1,j},S_{\pi(j-1),i})$, which highlights possibility of computing
the optimal simplification incrementally in a dynamic programming
algorithm.

It remains to consider the case of sequential collinear segments.  The
polyline $P'$ can be simplified into $P$ by merging sequential collinear
segments, effectively removing points of $P'$ without changing its shape. 
When joining two segments where $p'_i=q_j$, $p'_{i-1}$ and $p'{i+1}$ define
the regions as before but there is no longer a guarantee regarding
non-collinearity of $p'_{i-2}$ or $p'_{i+2}$ with respect to the other
points.  The points $q_{j-1}$ and $q_{j-2}$ are now collinear if and only if
either of them are entirely collinear to the relevant segments of $P$.  Our
check for equality is changed to a check for equality or collinearity.  We
examine the previous and next points of $P'$ that are not collinear to the
two segments $[p'_{i-1},p'_i]$ and $[p'_{i},p'_{i+1}]$.  We find such points
for every $p'_i$ in a preprocessing step requiring linear time and space, by
scanning the polyline for turns and keeping two queues of previous and
current collinear points.

\section{Optimal Crossing Measure Simplification}
\label{sect:method}

In this section we describe our dynamic programming approach to computing a
polyline $Q$ that is a subset of $P$ having minimum size $k$ conditional on
maximal crossing measure $\chi(Q,P)$.  We compute $\chi(S_{i,j},P_{i,j})$ in
batches, as described in the previous section.  Our algorithm will maintain
the best known simplifications of $P_{1,i}$ for all $i\in[1,n]$ and each of
the three possible labellings of the ends.  We refer to these paths as
$\mathcal(Q)_{\sigma,i}$ where $\sigma$ describes the labelling at $i$:
$\sigma=1$ for collinear, $\sigma=2$ for left, or $\sigma=3$ for right.

To reduce the space complexity we do not explicitly maintain the
(potentially exponential-size) set of
all simplifications $\mathcal{Q}_{\sigma,i}$.  Instead, for each
simplification corresponding to $(\sigma,i)$ we maintain:
$\chi(\mathcal{Q}_{\sigma,i},P_{1,i})$ (initially zero); the size of the
simplification found $|\mathcal{Q}_{\sigma,i}|$ (initially $n+1$); the
starting index of the last segment added $\beta_{\sigma,i}$ (initially
zero); and the end labelling of the best simplification that the last
segment was connected to $\tau_{\sigma,i}$ (initially zero).  The
initial values described represent the fact that no simplification is yet
known.  The algorithm begins by setting the values for the optimal identity
simplification for $P_{1,1}$ to the following values (note $\sigma=1$).

\begin{eqnarray*}
\chi(\mathcal{Q}_{1,1},P_{1,1}) = 0 \\
|\mathcal{Q}_{1,1}| = 1 \\
\beta_{1,1} = 1\\
\tau_{1,1} = 1
\end{eqnarray*} 

A total of $n-1$ iterations are performed one for each $i\in[1,n-1]$ where a 
batch of segments $S_{i,j} : i<j\leq n$ is each considered in a possible 
simplification ending in that segment.  Each iteration begins with the 
set of simplifications $\{\forall \sigma, \mathcal{Q}_{\sigma,\ell} : \ell \leq i\}$ being optimal,
with maximal values of $\chi(\mathcal{Q}_{\sigma,\ell},P_{1,\ell})$ and
minimum size $|\mathcal{Q}_{\sigma,\ell}|$ for each of the specified
$\sigma$ and $\ell$ combinations.  The iteration proceeds to calculate the
crossing numbers of all segments starting at $i$ and ending at a later
index $\{\chi(S_{i,j},P_{i,j}) | j\in (i,n], \}$, using the
method from Section~\ref{sect:measure}.  For each of the segments $S_{i,j}$
we compute the sidedness of both the end at $j$ ($\sigma'_j$) and the start at
$i$ ($\upsilon'_j$).  Using $\upsilon'_j$ and all values of $\{\sigma :
\beta_{\sigma,i} \geq 0\}$ it is possible to compute
$\eta(\mathcal{Q}_{\sigma,i},S_{i,j})$ using just the labellings of the two
inputs (see Table~\ref{tab:etapsi}).  It is also possible to determine
the labelling of the end of the concatenated polyline
$\psi(\sigma,\sigma'_j)$ using the labelling of the end of the previous
polyline $\sigma$ and the end of the additional segment $\sigma'_j$ (also
shown in Table~\ref{tab:etapsi}).
 
 \begin{table}[ht]
	\begin{minipage}[b]{0.45\linewidth}\centering
	\begin{tabular}{|c|ccc|}
	\hline  \multirow{2}{*}{$\eta(\beta_{\sigma,i},\upsilon'_j)$} & \multicolumn{3}{|c|}{$\beta_{\sigma,i}$} \\
	   & 1 & 2 & 3  \\  
	\hline
	$\upsilon'_j$ & & &  \\
	1 & 0 & 0 & 0  \\
	2  & 0 & 0 & 1  \\
	3  & 0 & 1 & 0  \\
	\hline
	\end{tabular}
	\end{minipage}
	\begin{minipage}[b]{0.45\linewidth}
	\centering
	\begin{tabular}{|c|ccc|}
	\hline  \multirow{2}{*}{$\psi(\sigma,\sigma'_j)$} & \multicolumn{3}{|c|}{$\sigma$} \\
	   & 1 & 2 & 3  \\  
	\hline
	$\sigma'_j$ & & & \ \\
	1  & 1 & 2 & 3  \\
	2  & 2 & 2 & 2  \\
	3  & 3 & 3 & 3  \\
	\hline
	\end{tabular}
	\end{minipage}
	\caption{Tables defining the computation of additional crossings $\eta(\beta_{\sigma,i},\upsilon'_j)$ due to concatenation and the end labelling of the concatenated polyline $\psi(\sigma,\sigma'_j)$}
	\label{tab:etapsi}
	\end{table}
 
With these values computed, the current value of
$\chi(\mathcal{Q}_{\psi(\sigma,\sigma'_j)})$ is compared to
$\chi(\mathcal{Q}_{\sigma,i}) + \chi(S_{i,j},P_{i,j}) +
\eta(\beta_{\sigma,i},\upsilon'_j)$ and if the new simplification has a
greater or equal number of crossings crossings then we can compute:
 
\begin{eqnarray*}
\chi(\mathcal{Q}_{\psi(\sigma,\sigma'_j),j},P_{1,j}) &=&  \chi(\mathcal{Q}_{\sigma,i}) + \chi(S_{i,j},P_{i,j}) + \eta(\beta_{\sigma,i},\upsilon'_j)\\
|\mathcal{Q}_{\psi(\sigma,\sigma'_j),j}| &=& |\mathcal{Q}_{\sigma,i}| + 1 \\
\beta_{\psi(\sigma,\sigma'_j),j} &=& i\\
\tau_{\psi(\sigma,\sigma'_j),j} &=& \sigma
\end{eqnarray*}

Correctness of this algorithm follows from the fact that each possible
segment ending at $i+1$ is considered before the $(i+1)$-st iteration.  For
each segment and each labelling, at least one optimal polyline with that
labelling and leading to
the beginning of that segment must have been considered, by the inductive
assumption.  Since the number of crossings in a polyline
only depends on the crossings
within the segments and the labellings where the segments meet,
the inductive hyopthesis is maintained through the $(i+1)$-st iteration.
It was also trivially true in the basis case $i=1$.
With the exception of computing the crossing number for all of the segments,
the algorithm requires $O(n)$ time and space to update the remaining
information in each iteration.  The final post-processing step is to determine
$\sigma_{max} = \arg\max_{\sigma} \chi(\mathcal{Q}_{\sigma,n})$,
finding the simplification that has the best crossing number.
We use the $\beta$ and $\tau$ information to reconstruct
$\mathcal{Q}_{\sigma_{max},n}$ in $O(k)$ remaining time.

The algorithm requires $O(n)$ space in each iteration and $O(n\log^2 n)$
time per iteration to compute crossings of each batch of segments dominates
the remaining time per iteration.  Thus for simple polylines
$\mathcal{Q}_{\sigma_{max},n}$ is computable in $O(n^2\log^2 n)$ time and
$O(n)$ space and for monotonic polylines it is computable in $O(n^2\log n)$
time and $O(n)$ space.

\section{Results and Smooth Shape Approximation}
\label{sect:monotonicresults}

Our goal was to approximate shape (and noise) in a parameterless fashion.  In
this section we present results of applying the simplification to monotonic
data with and without noise.  We then describe how we perform a
parameterless \emph{smooth boostrap-like} operation, and give results for the
median and the 95\% confidence intervals (i.e.,~error approximations).  We
conclude by showing the results applied to a spectrum acquired from a
Fourier transformed infrared microscope.

Our first point set is given by $p=(x,x^2+10\cdot \sin x)$ for 101 equally
spaced points $x\in[-10,10]$.  The maximal-crossing simplification for this
point set has 5 points and 7 crossings.  We generated a second point set by
adding standard normal noise generated in \textsc{Matlab} with
\textit{randn} to the first point set.  The maximal-crossing simplification
of the data with standard normal noise has 19 points and a crossing number
of 54.  We generated a third point set from the first by adding
heavy-tailed noise consisting of standard normal noise for 91 data points
and standard normal noise multiplied by ten for the remaining ten points. 
The maximal-crossing simplification of the signal contaminated by
heavy-tailed noise has 20 points and a crossing number of 50.  These results
are shown in Figure~\ref{fig:cleanresults}.
	
\begin{figure}
\subfloat{
	\includegraphics[scale=0.32]{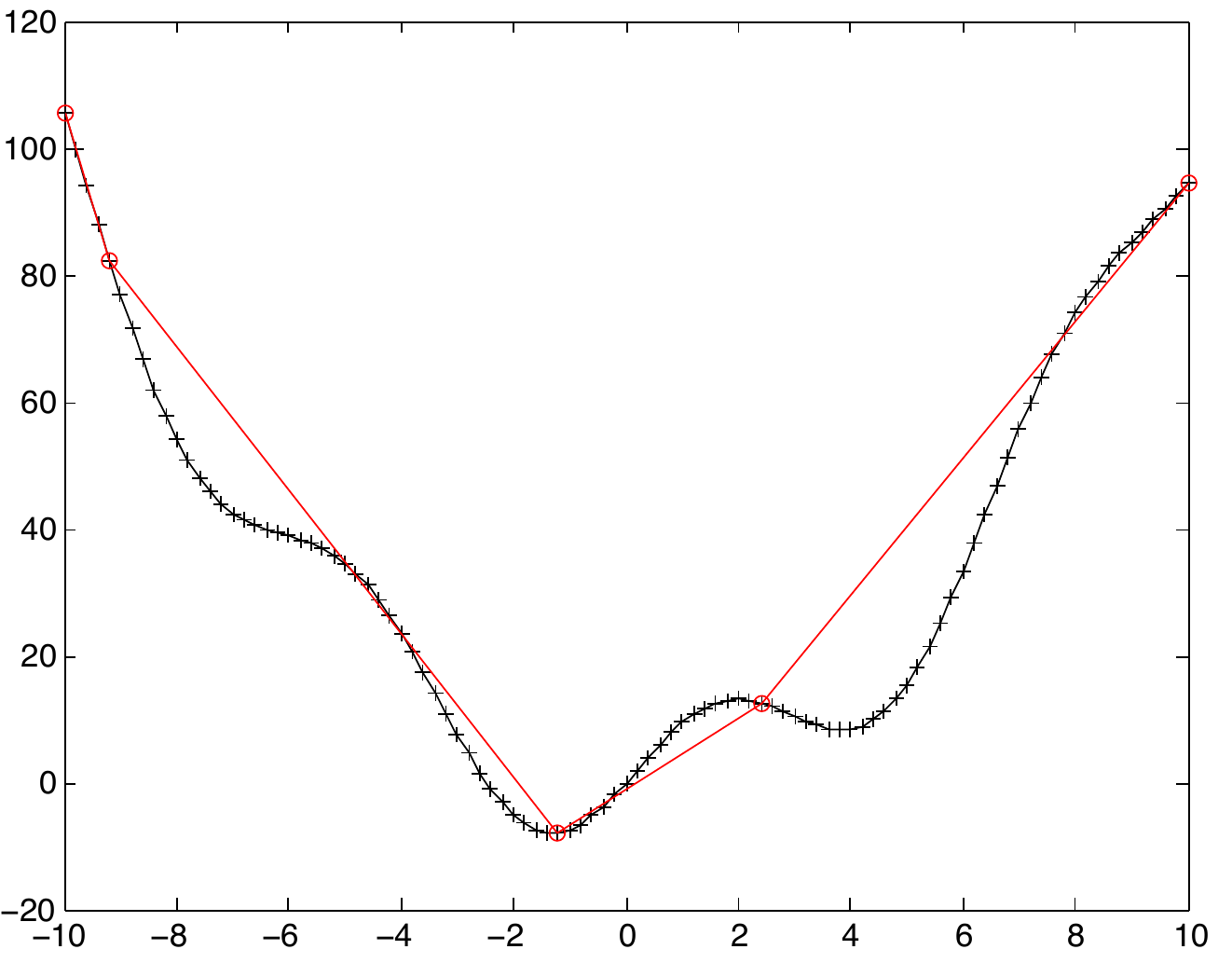}
}
\qquad
\subfloat{
	\includegraphics[scale=0.32]{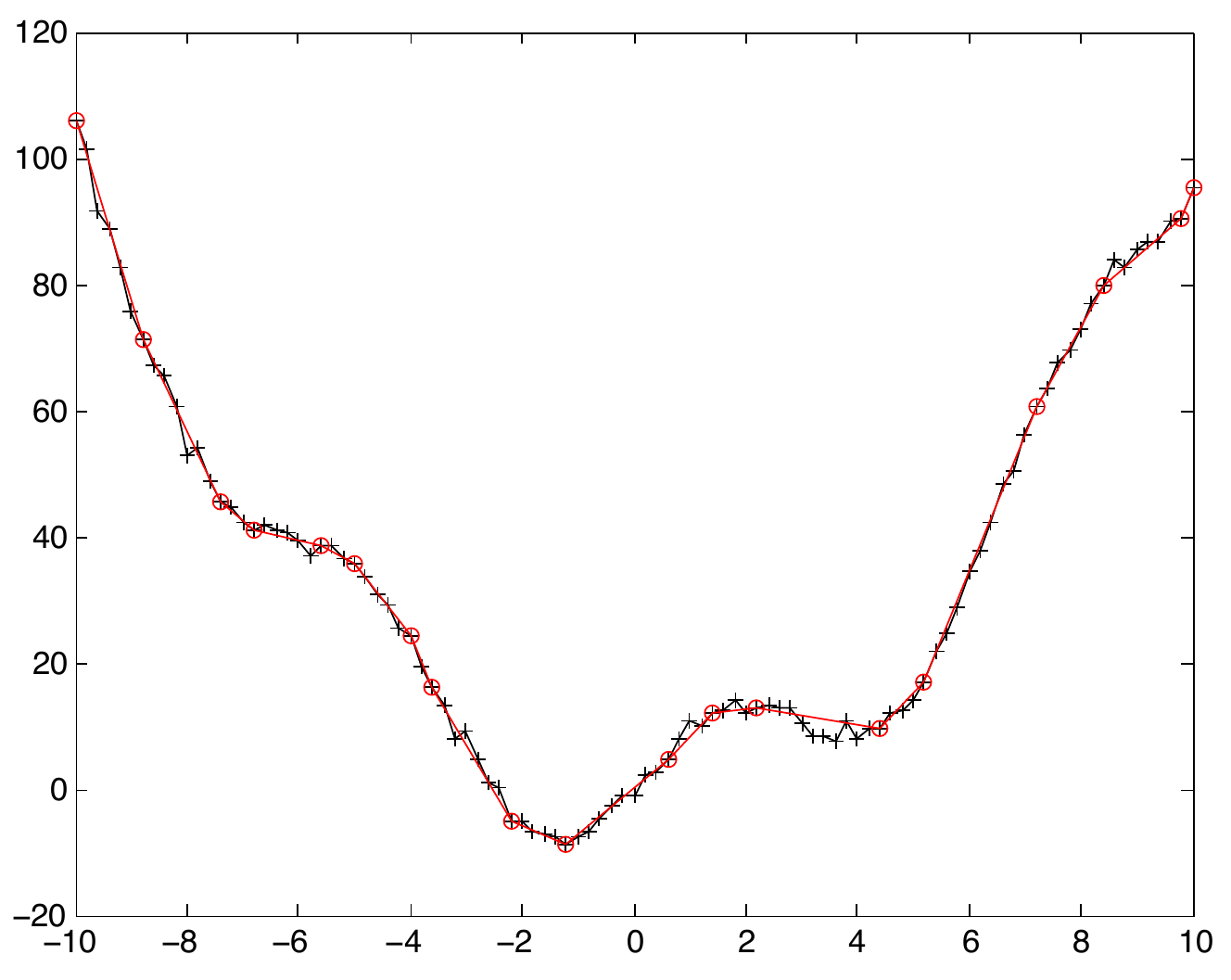}
}
\qquad
\subfloat{
	\includegraphics[scale=0.32]{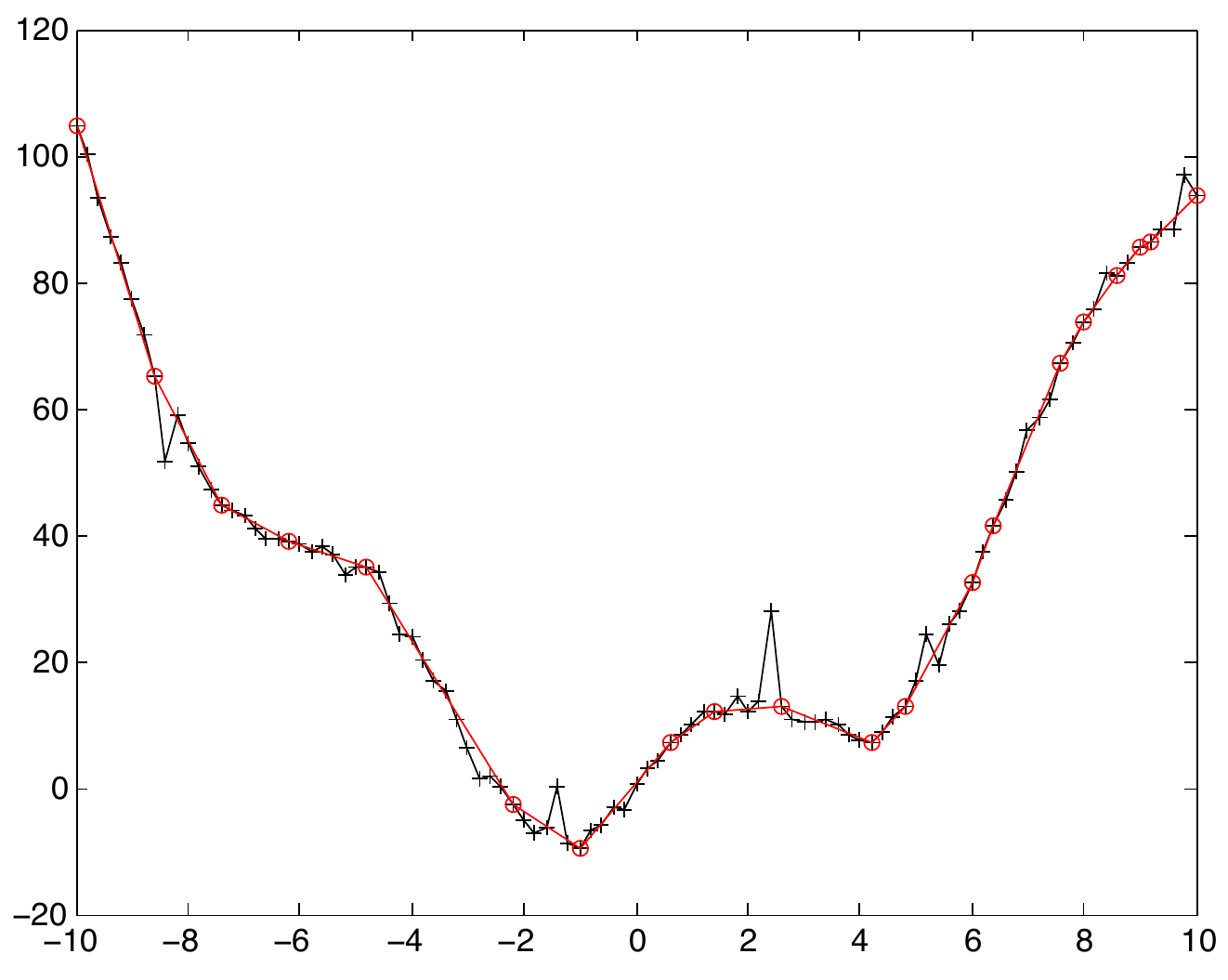}
}
\caption{The optimal crossing path for $p=(x,x^2+10\cdot \sin x)$, without
noise, with standard normal noise and with heavy-tailed (mixed gaussian)
noise.}
\label{fig:cleanresults}
\end{figure}

As can be seen in Figure~\ref{fig:cleanresults}, the crossing-maximization
procedure gives a much closer approximation to the signal when there is some
nonzero amount of noise present to provide opportunities for crossings.  We
might expect that in the case of a clean signal, we could obtain a more
useful approximation by artificially adding some noise before computing our
maximal-crossing polyline.  However, to do so requires choosing an
appropriate distribution for the added noise, and we wish to keep our
procedure parameterless.

The residuals between the data and the optimal crossing approximation form a good
first approximation of the noise within the data.  If these residuals are not
zero-centered then their median should be subtracted to provide a
zero-centered distribution.  We can take the original data points and
subtract, at each point, a value selected uniformly at random with
replacement from the zero-centered residuals.  Then by finding the
maximal-crossing polyline of the resulting modified data, we obtain a
\emph{noise-based approximation}.  We repeat this procedure for different random
selections of which residuals to apply to which data points.  This
\emph{smooth bootstrap-like} approach is similar to ``smooth bootstrap''
estimation, which normally would use a parameter-based model of the error
\cite{wilcox:2010fk}.  Our procedure is parameterless.  Repeated evaluation
of noise based approximations produce multiple $y$ values for all $x$ values. 
Using the median of these and finding the 5th and 95th percentiles results
in an approximation of signal and noise after relatively few iterations.  Results
from 90 iterations of this calculation applied to a Fourier transformed
infrared spectrum are shown in Figure~\ref{fig:ftir}.

\begin{figure}
\subfloat{
		\includegraphics[scale=0.5]{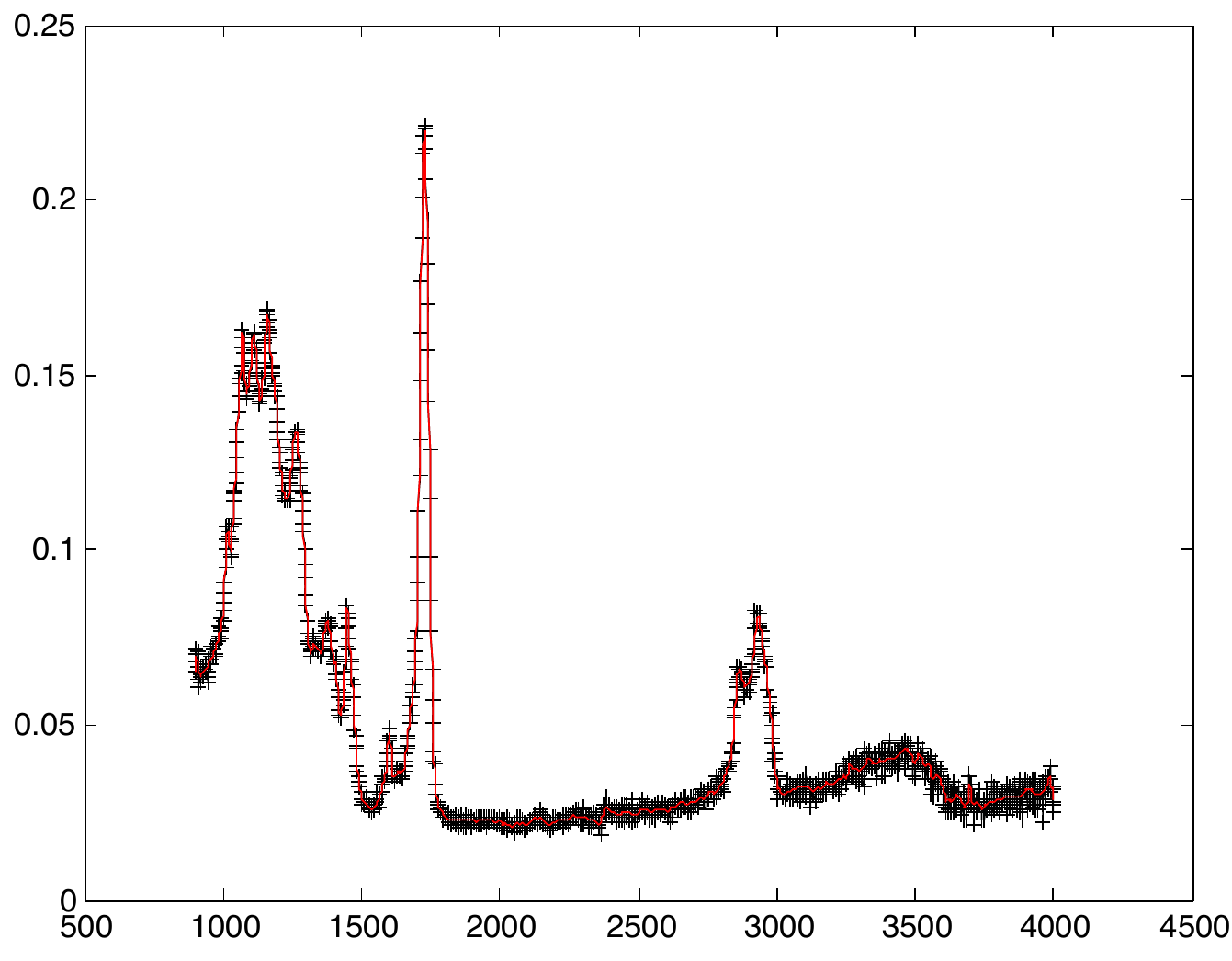}
	}
	\qquad
	\subfloat{
		\includegraphics[scale=0.5]{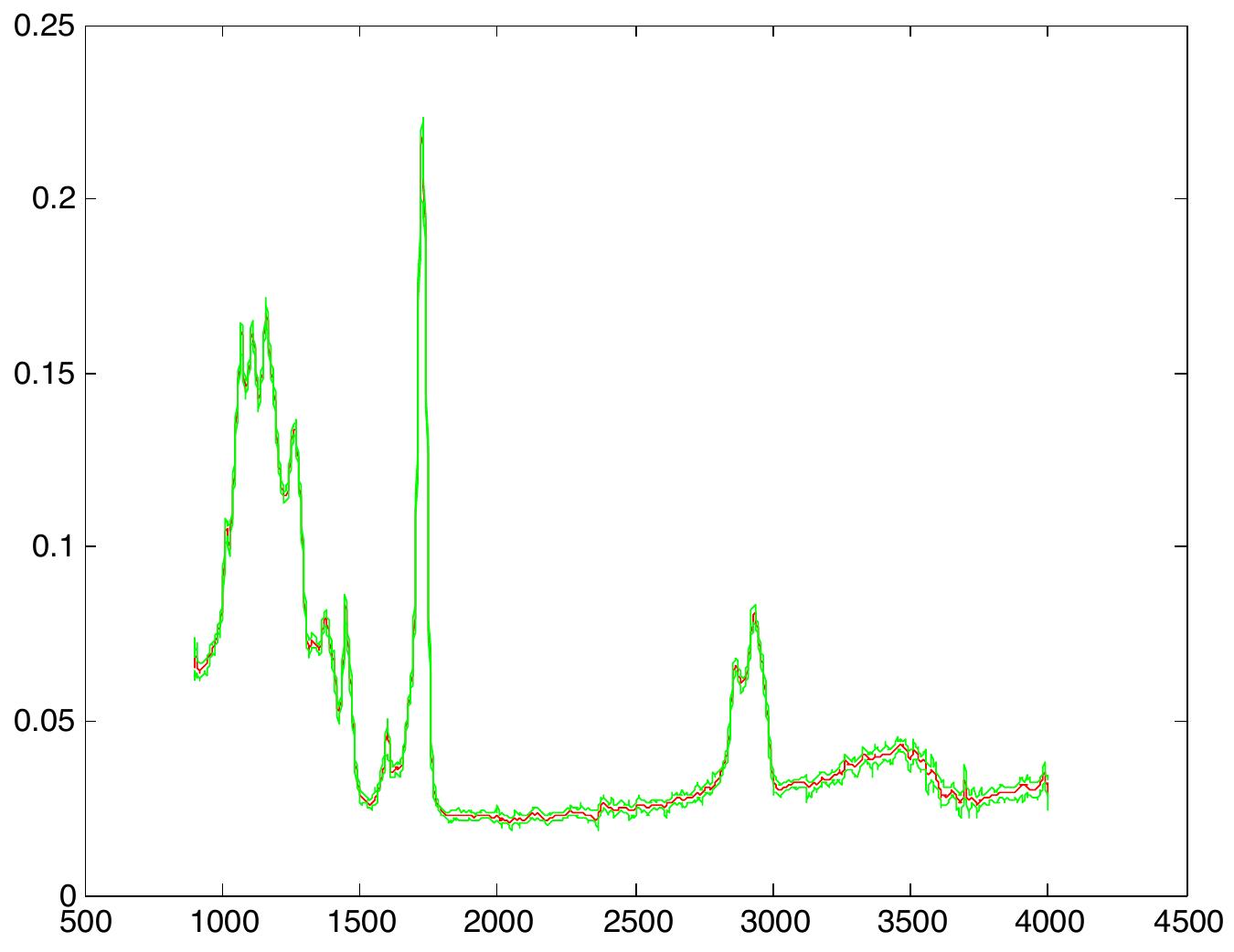}
	}
\caption{The data is in black (with crosses) and contains 1610 data points), the median bootstrap approximation is in red and the green lines are the bootstrapped approximations of the 5 and 95 percentiles.}
\label{fig:ftir}
\end{figure}

\section{Discussion and Conclusions}
\label{sect:discconc}

The optimal crossing measure simplification is robust to small changes
of $x$- or $y$-coordinates of any $p_i$ when the points are in general
position. This robustness can be seen by considering that the crossing
number of every simplification depends on the arrangement of lines
induced by the line segments, and any point in general position (by
definition) can be moved some $\epsilon$ without affecting the
combinatorial structure of the arrangement. The simplification is also
invariant under affine transformations because these too do not modify
the combinatorial structure of the arrangement. In the case of
$x$-monotonic polylines, the simplification possesses another useful
property: the more a point is an outlier, the less likely it is to be
included in the simplification. In the limit, increasing the
$y$-coordinate of any point $p_i$ to infinity ($x$-monotonicity
remains unchanged) will remove $p_i$ from the simplification. That is,
if $p_i$ is initially included in the simplification, then once $p_i$
moves sufficiently upward, the two segments of the simplification
adjacent to $p_i$ cease to cross any input segments in $P$.

We discuss an additional improvement achievable by bounding sequence lengths. 
If a parameter $m$ is chosen in advance such that we require that the
longest segment considered can span at most $m-2$ vertices, then with the
appropriate changes the algorithm can find the minimum sized simplification
conditional on maximum crossing number and having a longest segment of
length at most $m$ in $O(nm\log^2m)$ time for simple polylines or $O(nm\log m)$
time for monotonic polylines, both with linear space.  Since long line
segments tend to be rare in good simplifications, we can set $m$ to a relatively
small value and still obtain good simplification results while significantly
improving speed.

Finally, this research presents a method to approximate the shape and noise
without an implied model for the data nor a need for parameters. 
Application of the shape and noise approximation in the monotonic (functional)
data case has shown promising results when used in conjunction with the
bootstrap method described here.
	
\bibliographystyle{plain}
\bibliography{DataSimplification}

\end{document}